%% file: main2.tex
\documentclass[12pt, a4paper]{article}
\usepackage{amsfonts}
\usepackage{amsmath}
\usepackage{latexsym}
\usepackage[dvips]{graphics}
\usepackage{pgf}
\usepackage{hyperref}
\usepackage{pdflscape}
\usepackage{multicol}
\usepackage{float}
\usepackage[german,english]{babel}

\usepackage{natbib}
\usepackage{authblk}

\usepackage[nolist]{acronym}

\newtheorem{theorem}{Theorem}[section]

\newenvironment{proof}{{\tt PROOF:}}{$\Box$\\}

\newcounter{defcounter}

\numberwithin{defcounter}{section}

\newcounter{excounter}

\numberwithin{excounter}{section}

\begin{document}

\begin{acronym}
\acro{AIC}{Akaike Information Criterion}
\acro{APB}{Average Percentage Bias}
\acro{ASC}{Alternative Specific Constant}
\acro{BIC}{Bayesian Information Criterion}
\acro{bME}{bivariate Mendell-Elston}
\acro{CCL}{Composite Conditional Likelihood}
\acro{CDF}{Cumulative Distribution Function}
\acro{CEF}{Conditional Expectation Function}
\acro{CL}{Composite Likelihood}
\acro{CLAIC}{Composite Likelihood Akaike Information Criterion}
\acro{CLBIC}{Composite Likelihood Bayesian Information Criterion}
\acro{CML}{Composite Marginal Likelihood}
\acro{c.p.}{ceteris paribus}
\acro{DCM}{Discrete Choice Models}
\acro{DGP}{Data Generating Process}
\acro{DM}{Decision Maker}
\acro{GHK}{Geweke-Hajivassiliou-Keane}
\acro{IC}{Information Criterion}
\acro{iid}{independent, identically distributed}
\acro{KLD}{Kullback-Leibler Divergence}
\acro{LLN}{Law of Large Numbers}
\acro{MA}{model averaging}
\acro{MACML}{Maximum Approximate Composite Marginal Likelihood}
\acro{MAE}{Mean Absolute Error}
\acro{MC}{Monte Carlo}
\acro{MCMC}{Markov Chain Monte Carlo}
\acro{ME}{Mendell-Elston}
\acro{ML}{Maximum Likelihood} 
\acro{MNP}{Multinomial Probit}
\acro{MOP}{German Mobility Panel}
\acro{MSE}{Mean Squared Error}
\acro{MSL}{Maximum Simulated Likelihood}
\acro{MVNCDF}{multivariate normal cumulative distribution function}
\acro{OLS}{Ordinary Least Squares}
\acro{PQD}{Positive Quadrant Dependence}
\acro{RE}{Random Effect}
\acro{RMSE}{Root Mean Squared Error}
\acro{RUM}{Random Utility Model}
\acro{RV}{Random Variable}
\acro{SEM}{Structural Equation Models}
\acro{SJ}{Solow-Joe}
\acro{SLLN}{Strong Law of Large Numbers}
\acro{ULLN}{Uniform Law of Large Numbers}
\acro{WLLN}{Weak Law of Large Numbers}
\acro{WTP}{Willingness to Pay}
 \end{acronym}

\title{Model selection and model averaging in MACML-estimated \ac{MNP} models}

\author[a]{Manuel Batram\footnote{Corresponding author; E-mail adress: Manuel.Batram@uni-bielefeld.de}}
\author[a]{Dietmar Bauer}

\affil[a]{\footnotesize{Department of Economics, Bielefeld University, Postfach 10 01 31, D-33501 Bielefeld, Germany.}}
\date{}

\maketitle

\begin{abstract}
This paper provides a review of model selection and model averaging 
methods for multinomial probit models estimated using the \ac{MACML} 
approach. The proposed approaches are partitioned into test based 
methods (mostly derived from the likelihood ratio paradigm), methods 
based on information criteria and model averaging methods. 
\\
Many of the approaches first have been derived for models estimated using maximum 
likelihood and later adapted to the composite marginal likelihood 
framework. In this paper all approaches are applied to the \ac{MACML} approach for estimation. The investigation lists advantages and disadvantages of the various methods in terms of asymptotic properties as well as computational aspects. We find that likelihood-ratio-type tests and information criteria have a spotty performance when applied to \ac{MACML} models and instead propose the use of an empirical likelihood test. 
\\
Furthermore, we show that model averaging is easily adaptable to \ac{CML} estimation and has promising performance w.r.t to parameter recovery. Finally model averaging is applied to a real world example in order to demonstrate the feasibility of the method in real world sized problems.

\end{abstract}
\section{Introduction}
The two most commonly used families of discrete choice models are the  multinomial logit (MNL) and the multinomial probit (\ac{MNP}) model. Between these two the \ac{MNP} offers better modeling flexibility at the expense of higher computational  costs. To alleviate the computational burden of \ac{MNP} estimation, the \ac{MACML} estimation approach combines \ac{CML} and an analytic approximation to allow for simulation-free, and fast (but approximate) estimation of \ac{MNP} models (see \citep{bhat2011}). Bhat and coworkers have shown that estimation in the \ac{MACML} framework is very fast (see \cite{cherchi2016}) and can be  applied even for complicated \ac{MNP} models (see e.g. \cite{kamargianni2015}).

Even though it is sometimes possible to specify a discrete choice model solely through theoretical reasoning, more often than not the researcher has to rely on model selection in order to arrive at a useful model. Despite \cite{bhat2011} as well as \cite{bhat2014} offering some limited advice, the literature on model selection for MNP models using the MACML estimation methodology is not extensive. Moreover, model averaging which has been proposed in a number of different contexts (see e.g. \cite{gao2016}, \cite{hjort2003}, \cite{hansen2007}, \cite{wan2014}) has not been dealt with in the context of \ac{CML} estimation.

Therefore this paper aims to stimulate discussion on this important topic by surveying proposals for model selection methods for \ac{MNP} models estimated using \ac{MACML}.\footnote{Note that in order to safe space we will sometimes violate the distinction between estimation method and model by abbreviating '\ac{MNP} models estimated using \ac{MACML}' as \ac{MACML} models.} We study those methods under the premise of a given, finite collection of competing models denoted as $\mathbb{M}$, 
from which the researcher selects one using a data-driven procedure. 

In the literature one finds two approaches for model selection. The first approach is based on a penalized goodness-of-fit function which is associated with each of the competing models. This function incorporates the trade-off between model fit and model complexity. In the context of \ac{MACML} this procedure is based on information criteria (IC) derived from the respective pseudo-likelihoods. 
The decision rule is that the model with the smallest, estimated value is retained. In an abstract sense this
partitions the space of all observations into regions where a particular model is chosen.

The alternative approach is based on repeated hypothesis testing wherein the final model choice is taken by repeatedly performing hypothesis tests of a restricted model versus an unrestricted model. Stepwise regression techniques are examples of this procedure. Again a partitioning of the observation space results.

The estimated parameter after model selection based on either approach therefore might be presented as 
\begin{equation*}
\hat{\theta}_j = \sum_{m \in \mathbb{M}} \hat{w}_m (Z) \hat{\theta}_j^m,
\end{equation*}
where $Z = (Y_n,X_n), n=1,...,N,$ denotes the observations and the weight is defined as 
\begin{equation*}
\hat{w}_m(Z)  = \begin{cases}
1,\quad \text{model m is selected by the decision rule based on observations Z}\\
0,\quad \text{otherwise}.
\end{cases}
\end{equation*}
Note that the weights are random as they are also 'estimated' from the data because the components of the model selection procedure depend on the given sample. Furthermore the decision rules only allow for the selection of a single model, which is equivalent to $\sum_{m \in \mathbb{M}} \hat{w}_m(Z) = 1$.

As an alternative to model selection in this paper we introduce model averaging methods for models estimated using the \ac{MACML} paradigm. The idea of \ac{MA} 
is to use the information from all models and not only the 'winning' model. The first motivation for doing so is the acknowledgment that model selection based on just one sample is subject to randomness. From the theoretical point of view, \ac{MA} avoids problems related to post-model-selection inference (see e.g. \cite[206ff]{claeskens2008} or \cite{leeb2005}). \ac{MA} is also appealing to the practitioner because it has been shown that predictions are better in terms of mean squared error when they are based on model averaging rather than model selection (see e.g. \cite{gao2016}). 

The estimated parameter after model averaging might be presented as,
\begin{equation}
\hat{\theta}_j^{MA} = \sum_{m \in \mathbb{M}} \hat{w}_m(Z) \hat{\theta}_j^m,
\end{equation}
where the weights fulfill $\sum_{m \in \mathbb{M}} \hat{w}_m (Z) = 1$. 
There are several model averaging procedures which differ in the way the weights $\hat{w}_m(Z)$ are justified and computed.

In this paper we present the adaptations of the concepts well known in  the context of maximum likelihood estimation to the \ac{MACML} case. In all cases we discuss  advantages and disadvantages of the various approaches focusing on the analytical and numerical properties. This is in particular of importance as model selection and model averaging requires the computation of a number of models and hence exacerbates computational issues already occurring for the estimation of a single model estimation. 

The organization of the paper is as follows. First, we succinctly review the \ac{MACML} framework and argue that the \ac{MACML} pseudo-likelihood differs from the standard \ac{CML} framework. In the second section we review test- and information criterion-based model selection and discuss computational challenges. As a new contribution this section introduces a test which is framed in empirical likelihood theory. This is followed by a simulation exercise comparing the various approaches on two demonstration examples.

Third, we introduce and discuss \ac{MA} for \ac{MACML} followed by a simulation-based comparison of the discussed methods. 
In section~\ref{case} we present an empirical example based on real choice data from the German Mobility Panel in order to demonstrate the feasibility of model averaging. We conclude the article with a summary of our results and discuss practical implications of our findings.

\section{The MACML approach}
\label{chap:macml}
The first building block of the \ac{MACML} method is CML estimation (for a general survey of this method see \cite{varin2011}). In this paper only the special case of the so called pairwise likelihood is considered, although many results generalize also to more general CMLs. For pairwise CML  
the full-likelihood representing the joint probability of the observation of $T$ choices $C_{nt}$ ($n=1,...,N, t=1,...,T$) by one individual $n$ is replaced by the product of the probabilities of all pairs of choices, 
 
\begin{flalign}
lcml(\theta) = \sum_{n=1}^N \log cml_{n}(\theta) &=\sum_{n=1}^N \sum_{t=1}^{T-1} \sum_{t'=t+1}^{T}  \log cml_{ntt'}(\theta) \nonumber \\ 
&= \sum_{n=1}^N \sum_{t=1}^{T-1} \sum_{t'=t+1}^{T} \log \mathbb{P}(C_{nt}=y_{nt}, C_{nt'}=y_{nt'} |X_n, \theta). \label{eq:probit}
\end{flalign}

Note that the functions $cml_{ntt'}(\theta)$ are valid marginal likelihoods such that  key properties of the likelihood framework continue to hold. For example it directly follows that the composite score, which we define as $\partial lcml(\theta) = s_N(\theta)$, is of zero mean at the true parameter $\mathbb{E} s_N(\theta_0)=0$. \\

In analogy to \ac{ML} estimation the \ac{CML} estimator is defined as 
$$
\hat{\theta} = \arg \max_{\theta} lcml(\theta). 
$$
Under suitable regularity conditions this estimator is consistent 
for $N \to \infty$ (see \cite[190ff]{molenberghs2005}), whenever $T$ is fixed. Furthermore, it is possible to show that under suitable assumptions the estimator is asymptotically normally distributed, where the asymptotic variance is given by the inverse of the  Godambe/sandwich information matrix, $G(\theta) = H(\theta) J(\theta)^{-1} H(\theta)$
where the so called sensitivity matrix $H(\theta)$ is defined as $H(\theta) = \mathbb{E}[-\partial s_N(\theta)]/N$ and the variability matrix $J(\theta)$ is the variance of the composite score $J(\theta) =  var(s_N(\theta))/N$.
This is a standard result for misspecified (in the sense of using a pseudo likelihood instead of the true one) likelihood models which is due to the fact that the information identity (also called the second Bartlett identity) $J(\theta) \not= H(\theta)$ does not hold (see e.g. \cite{white1982}). In general the \ac{CML} estimator therefore is less efficient than the \ac{ML} estimator. 

When \ac{CML}  is utilized to estimate \ac{MNP} models, each $cml_{ntt'}(\theta)$ involves the evaluation of the \ac{MVNCDF}: As is well known the \ac{MNP} model can be derived using an underlying random utility function 
assigning the utility
$$
U_{ntk} = X_{ntk}\beta + \tilde X_{ntk}\alpha_n + e_{ntk}
$$
to the choice of alternative $k=1,...,K$  with characteristics $X_{ntk}, \tilde X_{ntk}$ in the $t$-th decision of individual $n$. Here $\beta$ denotes fixed parameters while $\alpha_n$ -- assumed to be normally distributed with zero mean and variance $\Omega$ -- allows for heterogeneity between deciders. 
Furthermore $e_{nt:} = [e_{ntk}]_{k=1,...,K}$ denotes the \ac{iid} random errors assumed to be normally distributed with zero mean and variance $\Sigma$. Consequently using the vector $U_{nt:} = [U_{ntk}]_{k=1,...,K}$ we obtain that $[U_{nt:}',U_{nt':}']'$ is normally distributed. 
\\
It follows that the probability that $y_{nt}$ and $y_{nt'}$ are the choices (that is have maximal random utility) in the $t$-th and $t'$-th choice respectively is given by 
\begin{equation} \label{eq:CDF}
\Phi_{2K-2}(b[X_{nt:},X_{nt':},\theta] \, ; \textbf{0}, \textbf{R}([X_{nt:},X_{nt':},\theta])),
\end{equation}
where $\theta$ collects all parameters contained in $\beta, \Omega, \Sigma$ for appropriate functions 
$b[X_{nt:},X_{nt':},\theta]$ and $\textbf{R}([X_{nt:},X_{nt':},\theta])$ defining the upper limits of integration and the correlation matrix respectively. Furthermore $\Phi_{2K-2}$ denotes a $2K-2$ dimensional \ac{MVNCDF} corresponding to mean zero  and variance-covariance matrix $\textbf{R}$.

Whenever $K$ is large, \ac{CML} estimation of \ac{MNP} models is still computationally demanding and the \ac{MACML} approach by \cite{bhat2011}, therefore, complements pairwise \ac{CML} estimation with an analytic approximation for \ac{MVNCDF}s. \cite{bhat2011} proposes the use of the \ac{SJ} approximation \citep{solow1990,joe1995}, whose general idea is to factorize the multivariate normal distribution into a product of conditional distributions, which are in turn approximated by linear projections. For a three dimensional case of calculating the \ac{MVNCDF} for 
$u_j \le b_j, j=1,2,3$ 
where $u_j$ are standard normally distributed 
with correlation $\textbf{R}_{ij}$ between $u_i$ and $u_j$ we obtain (using $I_j = \mathbb{I} (u_j \le b_j)$): 

\begin{flalign}
\Phi_3( {\textbf{b}};
\textbf{0}, \textbf{R}) &=
\mathbb{P}(u_{1} \le  {b}_{1}, u_{2} \le  {b}_{2}) \mathbb{P}(u_{3} \le  {b}_{3} |
u_{1} \le  {b}_{1},u_{2} \le  {b}_{2}) \label{eq:SJ}\\
& = \Phi_2( {b}_{1},  {b}_{2};0,\textbf{R}_{12}) \mathbb{E}[I_3 | 
I_{1} = 1, I_{2} =1]\nonumber\\ 
& \approx \Phi_2( {b}_{1},
 {b}_{2};0,\textbf{R}_{12}) \hat{p}^{3|12}( {\textbf{b}}, \textbf{R}) \nonumber \\ &=:
\hat{P}^{SJ:3|12}( {\textbf{b}}, \textbf{R}), \nonumber
\end{flalign}
where $\hat{p}^{3|12}( {\textbf{b}}, \textbf{R})$ denotes the linear projection 
$$
\hat{p}^{3|12}( {\textbf{b}}, \textbf{R}) := \Phi( {b}_{3}) + \textbf{q}({\textbf{b}}, \textbf{R})'
\textbf{Q}({\textbf{b}}, \textbf{R})^{-1}[1- \Phi( {b}_{1}),1-\Phi( {b}_{2})]'. 
$$
Here the entries of $\textbf{q}({\textbf{b}}, \textbf{R})$ and $\textbf{Q}({\textbf{b}}, \textbf{R})$ 
contain covariances of $I_i$ with $I_j$, which 
are functions requiring the evaluation of univariate and bivariate \ac{MVNCDF} functions. 
The positive definite matrix $\textbf{Q}({\textbf{b}}, \textbf{R})$ has smallest eigenvalues bounded away from zero for
bounded ${\textbf{b}}$ if the same holds for $\textbf{R}$. 
\\
Note that the ordering of the components for the approximation is arbitrary but influences the approximation quality. It is known that averaging over all possible permutations of coordinates improves the approximation accuracy at the price of higher computational load.
Therefore typically only a fixed number (with one being a popular value) of random permutations are averaged. 
\\
Note that the approximation does not guarantee that the approximated choice probabilities for all choices sums to one. 
Furthermore there are no guarantees that the approximation $\hat{p}^{3|12}( {\textbf{b}}, \textbf{R})$ lies in $[0,1]$. This is typically ensured using some ad hoc interventions. 
For further details regarding the \ac{SJ} approximation see \cite{joe1995}.

Even though the \ac{SJ} approximation is an important building block of \ac{MACML} and ensures the comparatively fast estimation of complex \ac{MNP} models, it is important to note that its application alters the pseudo-likelihood and that \ac{MACML} estimation is, therefore, not equivalent to \ac{CML} estimation. To be more precise it is possible to show that due to the \ac{SJ} approximation\footnote{This is also true for several other analytic approximations and no particular problem of the \ac{SJ} approximation (for further details see \cite{batram2016}). Furthermore, this inconsistency is also present in \ac{MSL}-estimates whenever the number of simulations is finite (see e.g. \cite{lee1992}).} the \ac{MACML} estimators for data generated from an \ac{MNP} is asymptotically biased and not consistent (see \cite{batram2016}). 
To the best of our knowledge there are no general results that quantify or state bounds for the deviations between the \ac{MACML} pseudo-likelihood and the \ac{CML} pseudo-likelihood even in large samples, 
although the deviation has been shown to be 
small in all real world case studies performed so far. 

A different way to look at the stated inconsistency is to note that, if the approximated choice probabilities are normalized to sum to one, they encode a parameterized mapping of the characteristics encoded in $X_{ntk}, \tilde X_{ntk}$  onto the choice probabilities. 
It appears that one can show\footnote{This result is work in progress.}  that the \ac{MACML} estimator using the normalized probabilities 
can be used to consistently  estimate the corresponding underlying parameters for this mapping. In this sense the \ac{MACML} estimator regains all the usual properties of consistency and asymptotic normality for a slightly altered model.  

For model selection and averaging we will need some more notation in the following sections: Assume that the parameter vector $\theta$ is of dimension $d$ and that we partition this vector into $\theta = (\tau,\gamma)$, where $\gamma$ is $p$-dimensional and $\tau$ has dimension $d-p$. The vector $\tau$ contains parameters that are known to be present in the model while we are unsure whether the parameters assigned to the vector $\gamma$ should be contained in the model. 
Our interest for model selection thus lies on $\gamma$ and we want to test $H_0$: $\gamma = \gamma_0$. The elements of $\gamma_0$ will most often equal zero. Restrictions to other values such as 1 are possible for example for variances.  
\\
With respect to the full parameter vector we denote the estimator where $\gamma$ is constrained according to the null hypothesis by $\hat{\theta}^{\gamma_0} = (\hat{\tau}^{\gamma_0}, \gamma_0)$ and the unconstrained estimator is $\hat{\theta} = (\hat{\tau},\hat{\gamma})$. Finally, let $G_{\gamma \gamma}(\theta)$ be the $p \times p$ trailing submatrix of the information matrix and $s_{N,\gamma}(\theta)$ the $p$-dimensional trailing subvector of the score which contains only entries related to $\gamma$. Furthermore, $G^{\gamma \gamma}(\theta)$ denotes the corresponding $p \times p$ trailing submatrix of  $G(\theta)^{-1}$.

\section{Model selection}

As discussed above, the \ac{MACML} estimation is not equivalent to \ac{CML} estimation. However, we start this section by surveying methods for the \ac{CML} framework and then -- in the following section -- assess their applicability to \ac{MACML} estimation by simulation.

\subsection{Likelihood-ratio-type tests}
\label{chap:tests}
In this section we focus only on likelihood-ratio-type tests. While it is easy to adapted score-type and Wald-type tests for the use within the \ac{CML} framework, the Wald-type tests suffer from the known shortcomings (lack of invariance to reparametrization and elliptical confidence region) and the tests based on the score are numerically unreliable (see \cite[193f]{molenberghs2005}). 

The standard likelihood-ratio tests are inapplicable because the ratio of two composite likelihoods does not adhere to an asymptotic $\chi_p^2$ distribution under $H_0$ but instead is a linear combination of independent but weighted $\chi_1^2$ distributions. If we are interested in testing $H_0: \gamma = \gamma_0$, then

\begin{equation}
\label{eq:CLR}
CLR(\gamma) = 2[lcml(\hat{\theta}) - lcml(\hat{\theta}^{\gamma_0})] \overset{a}{\sim} \sum_{j=1}^{p} \lambda_{j} (K_{j})^2,
\end{equation}
where $\lambda_j$ are eigenvalues of
$(H^{\gamma\gamma}(\hat{\theta}^{\gamma_0}))^{-1}G^{\gamma\gamma}(\hat{\theta}^{\gamma_0})$ (as discussed those are $p \times p$ matrices which are evaluated under the null hypothesis) and $K_j$ are independent standard normal random variables. Just like the standard likelihood ratio test its \ac{CML} counterpart rejects $H_0$ whenever $CLR$ is large.

There are several different adjustments to $CLR$, which aim to facilitate likelihood-ratio-type testing in the \ac{CML} framework. Due to space constraints we will only focus on the adjustments which have shown promising results in previous studies (for a more general overview and simulation results see \cite{pace2011} or \cite{cattelan2016}).\footnote{Note, however, that the simulation studies here usually focus on simple models like estimation of moments from the multivariate normal or Gaussian random fields. As those models are not a good surrogate for \ac{MACML}-estimated \ac{MNP} models  we provide our own simulations in section~\ref{chap:simu}.} Those adjustments either try to alter $CLR$ so that it is approximately $\chi^2$ distributed or match certain moments of the $\chi^2$ distribution.

The moment matching adjustments are motivated by the observation that if $p=1$ the eigenvalue is a simple ratio $\lambda_1 = J_{\gamma\gamma}(\hat{\theta}^{\gamma_0}) / H_{\gamma\gamma}(\hat{\theta}^{\gamma_0})$ and, therefore, $cCLR_1(\gamma) = CLR_1(\gamma)/\lambda_1$ is asymptotically $\chi_1^2$. For more than one parameter this adjustment has the form,\begin{equation}
cCLR_{1}(\gamma) = \frac{CLR(\gamma)}{\omega} \overset{a}{\sim} \chi_{p}^2,
\end{equation}
where $\omega = \sum_{k=1}^p\lambda_k/p$. This adjustment is equivalent to match the first moment of the $\chi^2$ distribution. A more sophisticated adjustment, which is designed to match the first and second moment, was proposed in \cite{varin2008},
\begin{equation}
cCLR_{2}(\gamma) = \frac{CLR(\gamma)}{\kappa} \overset{a}{\sim} \chi_{\nu}^2,
\end{equation}
where $\kappa = \sum_{j=1}^p\lambda_j^2 / \sum_{j=1}^p\lambda_j$ and $\nu = (\sum_{j=1}^p\lambda_j)^2 / \sum_{j=1}^p\lambda_j^2$. It is well known that this adjustment might be inaccurate because it only matches the first two moments. 
\\
Both the calculation of $cCLR_1(\gamma)$ and $cCLR_2(\gamma)$ are computationally only slightly more expensive than CML optimization: Only the eigenvalues $\lambda_i$ must be calculated based on estimates of $H(\hat \theta)$ and $G(\hat \theta)$ which demand one extra pass through the likelihood calculations (if the corresponding quantities are not stored in the last gradient descent step).   
\\
Another adjustment was first proposed by \cite{chandler2007} and later modified by \cite{pace2011} in order to be invariant to reparametrization. In contrast to the previous method this adjustment does not match moments but tries to ensure that the corresponding $CLR$ is asymptotically $\chi_p^2$ distributed,
\begin{equation*}
cCLR_3(\gamma) = \frac{s_{N,\gamma}(\hat{\theta}^{\gamma_0})' H^{\gamma\gamma}(\hat{\theta}^{\gamma_0})  [G^{\gamma\gamma}(\hat{\theta}^{\gamma_0})]^{-1} H^{\gamma\gamma}(\hat{\theta}^{\gamma_0})
s_{N,\gamma}(\hat{\theta}^{\gamma_0}) }{s_{N,\gamma}(\hat{\theta}^{\gamma_0})' H^{\gamma\gamma}(\hat{\theta}^{\gamma_0})^{-1} s_{N,\gamma}(\hat{\theta}^{\gamma_0})} CLR(\gamma)  \overset{a}{\sim} \chi_{p}^2.
\end{equation*}

Note that the numerator is the \ac{CML} version of a score test statistic but as for example explained in \cite{cattelan2016} $cCLR_3$ does not inherit its numerical instabilities. It is important to note that all those adjustments rely on estimators for the Godambe information matrix. Again the computation requires little extra work. 

As an alternative we introduce another test, which has the appeal that there is no need to compute the Godambe information matrix. This test is based on the empirical likelihood framework, which in the context of \ac{CML} estimation was first used by \cite{lunardon2013} to derive confidence regions for \ac{CML} estimates. Using general results from \cite{qin1994} it is straightforward to also introduce a likelihood-ratio-like test for \ac{CML}. The only property needed is that the composite score is an unbiased estimation equation, which it is under standard conditions as discussed in section~\ref{chap:macml}, and Corollary 5 from  \cite{qin1994}. Hence we do not provide any proofs.
\\
In general, this approach comes with two difficulties: First, regarding theoretical properties empirical likelihood methods are known for their slow convergence to their respective asymptotic distributions, which might impose problems for small sample sizes. Second there is a need for additional computations. The essential building block of empirical likelihood theory is,
\begin{equation*}
l_E(\theta) = 2 \sum_{n=1}^N \log[1 + \psi's_n(\theta) ],
\end{equation*}
with $s_n(\theta) = \partial \log cml_n(\theta)$, where $\psi = \psi(\theta)$ is a d-dimensional vector chosen to satisfy,
\begin{equation}
\label{eq:psi}
\frac{1}{N} \sum_{n=1}^{N}\frac{s_n(\theta)}{1+\psi's_n(\theta)} = 0.
\end{equation}
While the individual gradients are available from the \ac{CML} estimation we need to compute the Lagrange multiplier $\psi$ in order to calculate $l_E(\theta)$. The best way to determine $\psi$ is to use the strategy outlined in \cite{owen1990} and restate the root finding problem as a minimization problem (see \citep[104ff]{owen1990}). Note that for practical computations (\ref{eq:psi}) should be 'numerical zero'. In our simulation experiments even with extremely low (gradient and step) tolerances it was rare to find a solution which fulfilled (\ref{eq:psi}) exactly. Our first attempts to use direct (multivariate) root finding methods failed (e.g. \textit{NAGs c05qb}) because the solutions provided by those algorithms were small but too far from zero and interfered with the tests performance.\footnote{How much deviation from (\ref{eq:psi}) can be tolerated without harming the test performance remains a practically relevant but open research question.}
Therefore we opted for solving the minimization problem. 

After $l_E(\theta)$ is computed the empirical likelihood ratio statistic for testing $H_0: \gamma ={\gamma_0}$ is (see \citep[307]{qin1994}),
\begin{equation*}
EL(\gamma) = 2l_E(\hat{\theta}^{\gamma_0}) - 2l_E(\hat{\theta}) \overset{a}{\sim} \chi_{p}^2.
\end{equation*}
Numerically this procedure is more demanding than the adjustments of the composite likelihood ratios: For $EL$ two Lagrange multipliers need to be computed by searching for the appropriate roots, which is in general faster than the computation of the sensitivity as well as the variability matrix.\footnote{We refrain from presenting computation times alongside our results because those depend strongly on the implementation as well as on the specification of the computer used for the computations. However, the appendix provides those numbers for our example implementation and computer system.}
\\
The expected computation time is the reason we did not include Bootstrap tests into our comparison. Even though the theoretical results in \cite{aerts1999} show that the parametric bootstrap leads to a consistent estimator for the distribution of pseudo-likelihood tests under the null-hypothesis, the specific form of the parametric bootstrap 
(see e.g. \cite{bhat2014}) is set up in a way that two models need to be estimated for every bootstrap sample.\footnote{As discussed in \citep[195]{molenberghs2005} the parametric bootstrap is expected to break down once the assumption regarding the distribution of the error terms is wrong. It might therefore be worthwhile to assess the performance of Bhats bootstrap because from the theoretical standpoint the use of the \ac{SJ}-approximation in the \ac{MACML} pseudo-likelihood is not compatible with the assumption of multivariate normal error terms in the utility function.} For a reasonable number of bootstrap samples this leads to very high computation times. It might still be worth to consider the bootstrap to draw inference from the final model but it is certainly not suited to be part of the model selection process for real world sample sizes.

All in all we have identified four different methods to perform likelihood-ratio-type tests in the \ac{CML} framework. Three methods ($cCLR_1$, $cCLR_2$, $cCLR_3$) are adjustments to the composite likelihood ratio (\ref{eq:CLR}), which rely on estimates of the Godambe information matrix. The fourth method  ($EL$), which is based on empirical likelihood arguments, utilizes only the individual score functions.

\subsection{Information Criteria}

In this section we discuss the use of information criteria for model selection. The two classical criteria are the \ac{AIC} and \ac{BIC}. The \ac{AIC} was developed as an estimator for the \ac{KLD} between two competing models. Naturally the model with minimal Kullback-Leibler divergence would be the model of choice. Because of this the AIC is designed to select the mimimum-\ac{KLD}-model. However, the \ac{AIC} is not an estimate of the \ac{KLD} (see \citet[p. 28ff]{claeskens2008}). 
\\
The \ac{BIC} is designed to select the model with the highest posterior probability when equal prior probability is assigned to each model and vague priors for the model parameters are used (see \citet[p. 286]{burnham2002}). The BIC is an approximation to this quantity with the appealing feature "that the specification of the prior completely disappears in the formula of BIC" (\citet[p. 81]{claeskens2008}).

Two important properties of \ac{IC} are efficiency and consistency. An \ac{IC} is deemed efficient if it selects the model with minimum prediction error. Consistency is defined as the property that an \ac{IC} asymptotically selects the model with the minimum \ac{KLD} and if that condition is satisfied by more than one model the model with the lowest number of parameters. This leads to the selection of the true model (if it is part of the model space) with probability going to one as $n \to \infty$. When assessing BIC and AIC with regard to efficiency, it can be shown that the AIC is efficient while the BIC is not (see \citet[p. 112]{claeskens2008}). On the other hand it can also be established that the BIC is consistent while AIC is not consistent.

Both of those classic \ac{IC} have been adopted for use within the \ac{CML} framework. The \ac{CLAIC} was introduced in \cite{varin2005} and is a variant of the Takeuchi \ac{IC} (TIC) (see e.g. \citep[43f]{claeskens2008}). The major difference between \ac{AIC} and TIC is that the latter is developed without the assumption that the true \ac{DGP} is amongst the candidate models. Given that within the \ac{CML} framework we are willing to make a working independence assumption with respect to the \ac{DGP}, it seems reasonable that the \ac{CLAIC} is closer to the TIC rather than the actual \ac{AIC}. The \ac{CLAIC} selects the model minimizing
  \begin{equation}
 - 2lcml(\hat{\theta}) + 2tr[\hat{J}(\hat{\theta})\hat{H}(\hat{\theta})^{-1}], 
  \end{equation}
  where $\hat{J}(\hat{\theta})$ and $\hat{H}(\hat{\theta})$ are consistent, first-order unbiased
  estimators for $J(\theta_0)$ and $H(\theta_0)$, respectively. When concerned with the \ac{CML} framework it is important to note that the objective function ($cml(\theta)$) is not necessarily a proper likelihood function. Therefore, \cite{varin2005} use a modified definition of the \ac{KLD} to derive the \ac{CLAIC}. This definition focuses on the marginal distributions, but also imposes that the compared \ac{CML} models are composed of marginals with the same dimension. \cite{ng2014} address some of the differences between the \ac{CLAIC} and normal \ac{AIC}. To do so they focus on nested models with tractable likelihoods and make use of the theory of local alternatives. Under those premises they can show that all other things equal the probability to pick the true models is smaller for the \ac{CLAIC} than for the \ac{AIC}. Furthermore, they illustrate that -- as expected -- the \ac{CLAIC} gets closer to the \ac{AIC} for rising dimensions of the marginals.
  
The \ac{BIC} has been adapted in analogy to the definition of \ac{CLAIC} for the \ac{CML} framework by \cite{gao2010}. The \ac{CLBIC}  selects the model minimizing,
  \begin{equation}
-2lcml(\hat{\theta}) + \log(n) tr[\hat{J}(\hat{\theta})\hat{H}(\hat{\theta})^{-1}]
  \end{equation}
  where $\hat{J}(\hat{\theta})$ and $\hat{H}(\hat{\theta})$ are consistent, first-order unbiased
  estimators for $J(\theta_0)$ and $H(\theta_0)$, respectively. \cite{gao2010} show that the \ac{CLBIC} as defined above is a consistent information criterion. Furthermore, they provide an extension for high-dimensional data which is irrelevant in the \ac{MNP} context. Finally note that \ac{CLAIC} and \ac{CLBIC} rely on the sandwich information matrix and that the computational burden is therefore similar to the $cCLR$ tests discussed in the previous section.
\\
Considering the definition of \ac{CLAIC} and CLBIC one notices a subtle dependence of the penalty 
term $tr[\hat{J}(\hat{\theta})\hat{H}(\hat{\theta})^{-1}]$ on the parameter 
$\theta$. As formulated above hence the penalty term is not guaranteed to be a 
strictly monotonous function of the model order. One could evaluate the penalty term always at the biggest model which would guarantee monotonicity.
However, as \ac{CLAIC} and CLBIC is intended to be used also for non-nested comparisons this definition is not suitable.

\section{Model selection: Comparison by simulation}
\label{chap:simu}
After we have discussed several model selection procedures we will assess their respective performance using a simulation exercise. The model under consideration is a  mixed panel \ac{MNP} model with five alternatives. Those alternatives are explained by five explanatory variables (drawn from independent standard normal distributions) and their respective parameters are assumed to be an instance of a multivariate normal distribution with  mean
$\textbf{b} = (1.5, -1, 2, 1, \beta)$ and covariance matrix $\Omega$. The error terms are assumed to be normally distributed with a mean of zero and a diagonal covariance matrix whose entries are fixed at 0.5. 

We generated data sets by drawing values of the vector $\textbf{b}_n$ and the error term from their respective distribution. Based on those values we calculated the utility and the chosen alternative is the one with the highest utility. The data sets have sample sizes 300, 500 or 1000 and in either case we generate 500 of those data sets for each setup to assess the properties of the various selection methods.\footnote{For each sample size a small double-digit number of the Monte Carlo samples was not used in the final analysis due to obvious non-convergence of at least one of the estimators.}

Our aim is to estimate the linear ($\textbf{b}$) as well as the covariance 
parameters. The covariance parameters are estimated using the corresponding Cholesky decomposition ($\Omega = LL'$). In order to find the optimum of the likelihood function we rely on the BFGS-algorithm provided by MATLABs \textit{fminunc} function. To ensure competitive computation times we have derived the respective analytic gradient but other than that relied on the default options. Following \citep{bhat2011a} we initialize the optimizer at the true values and use one random permutation per observation for the \ac{SJ} approximation, which stays the same during the optimization. Note that we always fit the unrestricted and restricted model and that we do not use a simple 'plug-in technique' to compute the restricted pseudo-likelihood/estimate.

In order to compute the \ac{CLAIC}, \ac{CLBIC} and most tests we need estimators for the Godambe information matrix. A simple estimate for the sensitivity matrix ($H(\theta)$) is given by the Hessian of $lcml(\theta)$ evaluated at the \ac{CML} estimate $\hat{\theta}$,
\begin{equation*}
\hat{H}(\hat\theta) = - \frac{1}{N} \sum_{n=1}^N \partial^2 lcml_n(\hat{\theta}).
\end{equation*}
An alternative estimator is derived by exploiting that the respective information identities hold for the marginal likelihoods (see \cite{lindsay2011}), 
\begin{equation*}
\hat{H_1}(\hat\theta) = \frac{1}{N} \sum_{n=1}^N  \sum_{t=1}^{T-1} \sum_{t'=1}^T  \partial lcml_{ntt'}(\hat{\theta}) \partial lcml_{ntt'}(\hat{\theta})',
\end{equation*}
where $lcml_{ntt'}(\theta) = \log cml_{ntt'}(\theta)$ is a pairwise likelihood of the n-th subject. 
\\
The estimator $\hat{H_1}(\hat\theta)$ relies only on gradients, it can be computed on the fly, while the numerical computation of $\hat{H}(\hat\theta)$ sometimes takes longer than the estimation of the corresponding model. Furthermore, like \cite[43f]{claeskens2008} we found that the estimates for $J(\theta)$ and $H(\theta)$ are subject to (at times severe) sampling variability and that this variance is in general lower for $\hat{H_1}$ compared to $\hat{H}$. We, therefore, base our analysis on the estimator $\hat{H_1}$.

For our setup and for \ac{MNP} models in general the sample size $N$ is certainly larger than the number of repeated observations $T$ and, therefore, the variability matrix $J(\theta)$ can be estimated by 
\begin{equation*}
\hat{J}(\hat\theta) = \frac{1}{N} \sum_{n=1}^N \partial lcml_n (\hat{\theta}) \partial lcml_n (\hat{\theta})'. 
\end{equation*}

Alternative methods to compute the Godambe information matrix are a Jackknife-based estimator (for the general idea see \cite[p. 302ff]{joe1997} and for an application see \cite{zhao2005}). Furthermore, it is possible to obtain estimates for $J$ and $H$ by simulation (see \cite{cattelan2016}). We leave those options for further research.

\subsection{Variable Selection}

\label{chap:select}
In this section we address the selection of variables. The setup is simple in that we simulate the addition of a fixed effect. We still assume that the other four variables are \ac{RE}s so $\textbf{b} = (1.5, -1, 2, 1, \beta)$ but the covariance matrix $\Omega$ is only $4 \times 4$. We only consider uncorrelated \ac{RE},

\begin{equation*}
\Omega=
\begin{bmatrix}
 2 & 0 & 0 & 0 \\
0 & 1.5 & 0 & 0\\
0 & 0 & 1 & 0\\
0 & 0 & 0 & 1.2\\
\end{bmatrix}.
\end{equation*}

The decision involves only one parameter, which is either restricted to zero or estimated from the data - $H_0: \textbf{b}_5 = 0$. In order to explore the power of the tests we vary the true $\beta$ from 0.1 to 0.5 by steps of 0.1. This is a simple task, which does not even require the use of a (composite) likelihood-ratio-test, but the results offer a first look into the performance of the various tests. Note that by definition $cCLR_1$, $cCLR_2$  and $cCLR_3$ are the same in the one parameter case.

The results are given in Table~\ref{tab:lin}. First note that we have included the naive composite likelihood ratio test, which is based on using $CLR$ (\ref{eq:CLR}) and a $\chi^2_1$ distribution without any corrections. As shown in the first column the size ($P(\text{test rejects } H_0| H_0)$) of this test is severely inflated. Even though the nominal size is supposed to be $0.05$ the empirical size is nearly 6 times larger. Furthermore, the power ($P(\text{test rejects } H_0| H_1)$) seems satisfactory which is to be expected because the test has a high rejection rate anyway. The size inflation seems to get worse for rising sample sizes.

\begin{table}[H]
\centering
\caption{Variable Selection: Empirical probability of rejecting $H_0$ at 0.05 confidence level, or -- for \ac{CLAIC} and \ac{CLBIC} empirical probability of selecting the larger model,  for various values of $\alpha$ each based on 500 simulated data sets}
\input{tab_linear.tex}
\label{tab:lin}
\end{table}

Regarding the various corrected tests, which are collectively presented as $cCLR$, our simulations reveal that the corrections work but that those tests still have an inflated size. Again, the nominal size is supposed to be $0.05$ the empirical sizes we observe are in the worst case nearly 3 times larger for $cCLR_1$, $cCLR_2$ and $cCLR_3$. All tests reject $H_0$ more often than theory would suggest.  We  observe that the size inflation gets more pronounced as the sample size gets larger. By noting that those tests are constructed by premultiplying $CLR$ (\ref{eq:CLR}) with a corrective term, this hints to the fact that the correction is not able to keep up with the inflation inherited from $CLR$. The power ($P(\text{test rejects } H_0| H_1)$) for each test is satisfying and improves for larger sample sizes. However, the second column reveals problems for small deviations in small data sets, which is to be expected. Finally, note that the power results are not adjusted for the difference between nominal and empirical size. 

For the the empirical likelihood test ($EL$), which is not based on a correction of (\ref{eq:CLR}), the empirical size is closer to the nominal size and not or only slightly inflated. Furthermore, the results suggest that the performance of the empirical likelihood test is in general superior compared to the correction-based tests because the power is higher or equal for all values of $\beta$. This dominance is especially pronounced for the sample size 1000, which is not surprising because as discussed in section~\ref{chap:tests} empirical likelihoods methods have the drawback of slow convergence.

The performance difference between the correction-based tests and the $EL$ test is also depicted in Figure~\ref{fig:ecdf} where we have plotted the empirical distribution function of the empirically observed values of the test statistics as well as the $\chi^2$ distribution they aim to approximate. It is clearly visible that the $EL$ closely follows the desired $\chi^2$ distribution while the other tests have too many large values. Even though  this is a one-dimensional example, where we only need to divide by the relevant eigenvalue to attain the desired distribution $\chi_1^2$ the estimation of this eigenvalue seems to be imprecise.

\begin{figure}%
{\caption{Empirical distribution function of the test statistics for sample size 1000 and the theoretical distribution of the corresponding $\chi^2$ distribution with one degree of freedom.}
\label{fig:ecdf}}%
\begin{tabular}{@{}r@{}} 
\includegraphics[width=\textwidth]{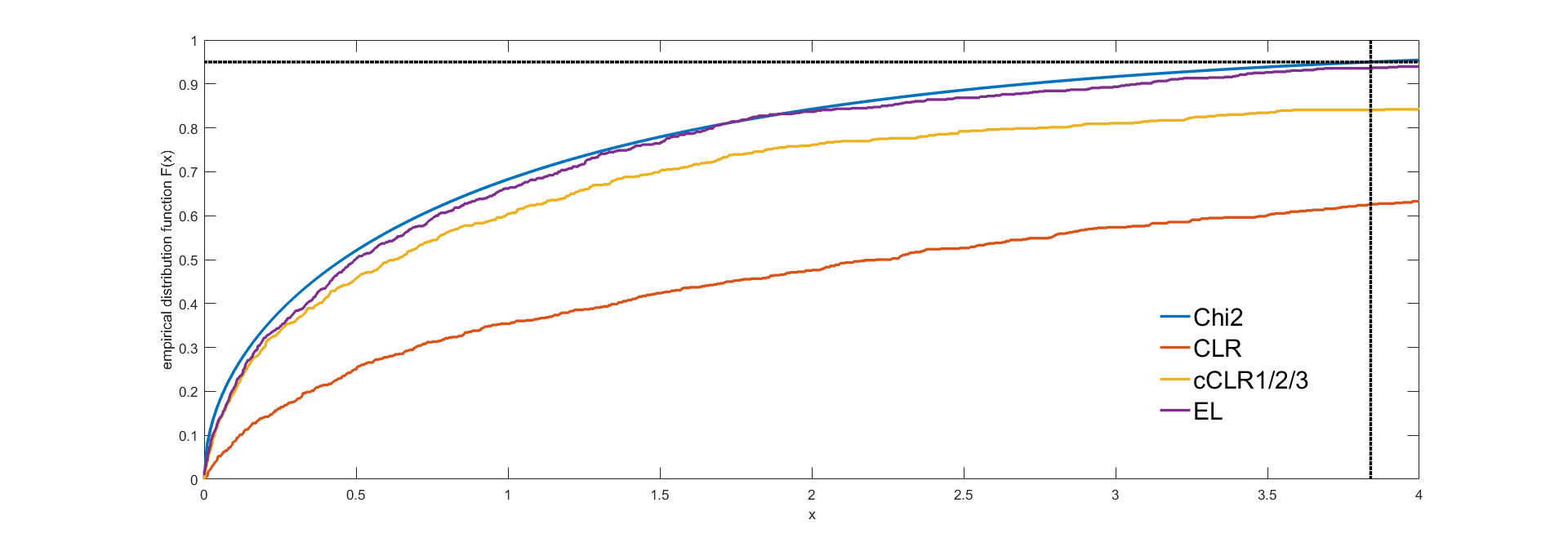}\\
\end{tabular}
\end{figure}

In Table~\ref{tab:lin} the last two rows for every sample size contain information regarding the performance of the information criteria. We choose the show the empirical probability that the larger model is selected which should be close to one for all values of $\beta$ except for the first column. We see that regardless of the underlying true model both criteria are highly in favor of the larger model. An interesting fact to note from the first column of Table~\ref{tab:cov} is that even the \ac{CLBIC}  is predominantly selecting the 'wrong' model just like the \ac{CLAIC} despite being a consistent \ac{IC}. With rising sample sizes the probability to select the 'true model' goes up, which is a reminder that this property is only adhered asymptotically.

\subsection{Choosing Covariance structures}

\label{chap:cov}
In this section we address the issue of choosing between different covariance structures for the \ac{RE}. More specifically we aim to test whether the off-diagonal entries of $\Omega$ are zero (diagonal covariance, $q=10$). For that reason $\beta$ is fixed as $\beta = 2$. In order to assess the empirical size of the test we first simulate 1000 data sets from a model with diagonal covariance ($\Omega_{I} = I_5$). We address the power of the tests using covariance matrices that mimic a Toeplitz matrices,

\begin{equation*}
\Omega_{\alpha} =
\begin{bmatrix}
 1 & \alpha & \alpha^2 & \alpha^3 & 0\\
\alpha & 1 & \alpha & \alpha^2 & 0\\
\alpha^2 & \alpha & 1 & \alpha & 0\\
\alpha^3 & \alpha^2 & \alpha & 1 &0\\
0& 0 & 0 & 0 & 1
\end{bmatrix}.
\end{equation*}

In order to represent various degrees of correlation between the \ac{RE}s we vary the true $\alpha$ from 0.1 to 0.4 by steps of 0.1. The decision involves the ten off-diagonal parameters of $\Omega$ 
which in the estimation are either all restricted to zero or all estimated from the data without restrictions.

Table~\ref{tab:cov} shows that for the current setup the performance of the tests differs from that of the last section. First, the problems of the naive test are even more pronounced than for the variable selection task as $H_0$ is almost always rejected. However, contrary to the results of the last section the size inflation of the corrected test ($cCLR_1, cCLR_2, cCLR_3$) is less severe but the difference between nominal and empirical size is, again, growing with the sample size.  Furthermore, while $cCLR_1$ and $cCLR_2$ are again almost equal in performance there is a performance difference in comparison to $cCLR_3$. The latter has an empirical size which is more inflated for small samples but the power is higher even when (for sample size 1000) the size inflation is equal to the one of $cCLR_1$ and $cCLR_2$. 

For this setup the $EL$ test suffers from size inflation, which is for the smallest sample size more severe than for the correction based test. As in the previous example the $EL$ test has the empirical size which is closest to the nominal level for a sample size of 1000. This serves as a reminder that empirical likelihood methods might not be adequate for small sample sizes. Furthermore, for sample size 1000 the power is lower than that of the correction-based tests at least for $\alpha = \{0.1, 0.2\}$. 

Regarding the \ac{IC}s it is clear that both \ac{CLAIC} and \ac{CLBIC} again favor the larger model, that is the unrestricted correlation structure, over the diagonal correlation matrix. This effect is extremely pronounced as the small model is almost never selected and illustrates that selecting the covariance structure using \ac{IC} will not work.

\begin{table}[H]
\centering
\caption{Covariance Structure: Empirical probability of rejecting $H_0$ at 0.05 confidence level, or -- for \ac{CLAIC} and \ac{CLBIC} empirical probability of selecting the larger model,  for various values of $\alpha$ each based on 500 simulated data sets}
\input{tab_cov.tex}
\label{tab:cov}
\end{table}

\section{Model averaging}
\label{chap:MA}

An alternative to model selection is model averaging (\ac{MA}). In the following we only refer to frequentist model averaging because we are concerned with \ac{MACML} estimation, a frequentist method. However, it is worth to note that there exists a large body of literature on Bayesian \ac{MA}.
\\
In \ac{MA} instead of selecting just one model, the estimates from several competing models are combined in a weighted average. 
In this section we assume that the models are labeled consecutively as 
$m = 1,2,...,M$. A model is characterized by setting some coordinates to 
prespecified values (often zero) and letting the other coordinates vary 
freely. The main case in this respect is regressor selection where setting
a coordinate to zero implies that the corresponding regressor does not 
influence the outcome. 
\\
The estimated $j$-th parameter after model averaging then might be written as
\begin{equation}
\hat{\theta}_j^{MA} = \sum_{m =1}^{M} \hat{w}_m (Z) \hat{\theta}_j^m,
\end{equation}
where the weight is defined 
such that $\sum_{m =1}^M \hat{w}_m (Z)= 1$ and where
$\hat \theta_j^m = \theta_{0,j}$ if the parameter is not contained in the model. 
Different versions of \ac{MA} methods use different weights $\hat{w}_m (Z)$. 

A simple class of \ac{MA} methods is based on information criteria where the weight of every model is derived from its respective \ac{AIC} or \ac{BIC} value. Those ad-hoc methods, which originated from \cite{buckland1997} aim to incorporate the uncertainty involved in model selection. The weights for a given sample are defined as

\begin{equation}
\label{eq:AICw}
\hat{w}_m (Z) = \frac{exp(0.5 IC_{m})}{\sum_{s \in M} exp(0.5 IC_{s})},
\end{equation}
where in our case $IC_{m}$ is either the \ac{CLAIC} or \ac{CLBIC} calculated for model $m$. Note that in order to simplify the calculations the $IC_{m}$ are often normalized to $\Delta(IC_{m}) = IC_{m} - \max_{s \in M} IC_{s}$, in either case the weights sum up to one by definition. Furthermore all weights  $\hat{w}_m (Z)$ are positive by definition. A theoretical justification for this model averaging strategy is provided in \cite{burnham2002} where it is shown that asymptotically even when those weights are computed from the \ac{AIC} they might be interpreted as (Bayesian) posterior probability that the model $m$ is correct (see section 6.4.5. in \citep{burnham2002}).

Alternatively asymptotically optimal \ac{MA} methods base the weight choice on the minimization of some criterion function, e.g. the asymptotic \ac{MSE} in the parameter or in prediction of certain focus quantities like conditional probabilities. These techniques typically are motivated within 
a local misspecification framework proposed in \cite{hjort2003}: 
The parameter vector $\theta = [\tau',\gamma']' \in {\mathbb R}^d, \gamma \in {\mathbb R}^p$ is partitioned into a set of parameters that are shared by all candidate models ($\tau$) and some additional parameters which only appear in some models ($\gamma$). Therefore, there is a range of models starting with a 'narrow model' where 
$\gamma = \gamma_0$ (often $\gamma_0=0$), which contains only the mandatory parameters and ending 
with the 'wide model', which is based on the full set of parameters. Within the local misspecification framework geared towards difficult to separate cases the data is assumed to be generated by the parameters $[\tau_0',\gamma_0' + \delta'/\sqrt{N}]'$ converging to $\theta_0$. In this setting the bias for not including a variable in the model is of the same magnitude as the squared error due to sampling variability, making the decision of whether to include the variable corresponding to $\gamma$ based on the available amount of data a hard decision. 
\\
\cite{wan2014}) show how to derive such a setup for Multinomial Logit Models in connection with maximum likelihood estimation, but  -- to the best of our knowledge -- no results are available neither for \ac{MNP} nor for \ac{CML} methods in general up to now.\\
Following \cite{hjort2003} an optimal \ac{MA} is based on the following insight:  

\begin{theorem} \label{thm:score}
Let $U_n = \partial lcml_n(\theta_0)/\partial \theta$
and let $\bar{U_n} = N^{-1} \sum_{n=1}^N U_n$. 
Further let $J = var_0(U_n) \in {\mathbb R}^{d \times d}$ (variance under $\theta_0$) and $H = -\mathbb{E}_0 \partial^2 lcml_n(\theta_0), H = [H_\tau,H_\gamma], H_\tau \in {\mathbb R}^{d \times (d-p)}, H_\gamma \in {\mathbb R}^{d \times p}$. 
\\
Let the model be such that the function $lcml_n(\theta)$ is three times continuously differentiable in a 
neighborhood of $\theta_0$ where all derivatives are dominated by functions with finite means under 
$\theta_0$. Furthermore the variables $U_n$ have finite fourth moments under $\theta_0$. 
\\
Then if the data is generated as iid draws under the sequence of local alternatives $[\tau_0',\gamma_0' + \delta'/\sqrt{N}]'$ we have:
\begin{equation} \label{equ:asyScore} 
\sqrt{N} \bar{U_n} \stackrel{d}{\to} U, \quad U \sim N(H_\gamma \delta, J)
\end{equation}
\end{theorem}

The theorem follows from Lemma 3.1. in \cite{hjort2003} as the assumptions directly imply assumptions (C2) and (C4) there. A slight difference occurs as in the current case the second derivative of the log-likelihood does not equal the negative variance of the score and thus different matrices $H$ and $J$ appear in the distribution of $U$ instead of only $J$. 
\\
This theorem can be applied to the current setting, if the number of choices is constant over individuals (generalizations to cases of unequal number of choices per individual are immediate). 
Note in this respect that using the SJ approximation to the \ac{MVNCDF} implicitly defines a model by providing a mapping between parameters and regressors onto choice probabilities. 
The corresponding model will be close to the \ac{MNP} but not identical. Consequently these two models and corresponding estimation methods need to be dealt with separately. 

\begin{theorem} \label{thm:MNP}
Assume that the data are generated by a \ac{MNP} with $T$ repeated choices for each individual, 
where the regressors $\| X_{ntk} \| \le M_X, \| \tilde X_{ntk} \| \le \tilde M_X$ are uniformly 
bounded. Furthermore the parameter set is compact and such that $\| \Sigma^{-1} \| \le M_\Sigma$. \\
In this case the conditions of Theorem~\ref{thm:score} hold for the 
\ac{MNP} model.
\\
The same conclusions hold for data generated according to model implied by the \ac{SJ} 
approximation.  
\end{theorem}
The proof of the theorem is based on the fact that the \ac{MVNCDF} is continuously 
differentiable in all its arguments (see \cite{Plackett} for derivatives with respect to entries 
in the correlation matrix; derivatives with respect to upper integration limits are obviously 
given by the corresponding PDF). Thus also higher order derivatives are continuously 
differentiable. The bounds on the regressors and the eigenvalues of the covariance matrix 
imply bounds on the upper limits as well as the eigenvalues of the correlation matrices occurring in \eqref{eq:CDF}. 
For the \ac{SJ} case also only Gaussian CDFs and PDFs occur and hence the result is immediate. Details are omitted. 
\\
Standard asymptotic expansions of the score in combination with the mean value theorem then provide the following consequence of Theorem~\ref{thm:score}: 
\begin{theorem} \label{thm:est} 
Under the conditions of Theorem~\ref{thm:MNP} let $\hat \theta_m$ denote the estimator (over the compact parameter set $\Theta_m$ of which $\theta_0$ is an interior point) of the model $m$ defined via the fact that in $\gamma$ for the estimation certain entries are chosen according to $\gamma_0$ while the remaining are estimated. Then 
$$
\sqrt{N} 
(\hat \theta_m - \theta_0) \stackrel{d}{\to} \Lambda_m U 
$$
for some matrix $\Lambda_m \in {\mathbb R}^{d \times d}$. 
\\
If furthermore the function $\mu(\theta) \in {\mathbb R}$ is continuously differentiable at $\theta_0$ with gradient $\partial_\theta \mu(\theta_0) = [\partial_\tau \mu(\theta_0),\partial_\gamma \mu(\theta_0)]$ then
$$
\sqrt{N}(\mu(\hat \theta_m)-\mu( [\tau_0',\gamma_0' + \delta'/\sqrt{N}]')) \stackrel{d}{\to}
\partial_\theta \mu \Lambda_m U - \partial_\gamma \mu \delta
$$
and hence in the limit is normally distributed with mean $\lambda_m \delta
=(\partial_\theta\mu \Lambda_m H_\gamma - \partial_\gamma  \mu) \delta$ and variance $ \partial_\theta\mu \Lambda_m J 
(\partial_\theta\mu \Lambda_m)'$. Consequently $N$ times the mean
square error $MSE_m$ for the model $m$ 
asymptotically equals 
$$
N MSE_m = \lambda_m  \delta \delta' 
\lambda_m ' + \partial_\theta\mu \Lambda_m J 
(\partial_\theta\mu \Lambda_m)'.
$$
\end{theorem} 
\begin{proof}
The proof follows standard mean value expansions: 
As $\hat \theta_m$ maximizes the CML function at an interior point of 
the parameter set, its derivative at the estimator is zero:
$$
0 = \partial_{\theta_m} lcml(\hat \theta_m) = 
\partial_{\theta_m} lcml(\theta_0) + 
\partial_{\theta_m \theta_m}^2 lcml(\bar \theta) ( \hat \theta_m - \theta_0)
$$
Here $\bar\theta$ denotes an intermediate value. Standard theory implies that $\partial_{\theta_m \theta}^2 lcml(\bar \theta) \to -\pi_m H \pi_m'$ where $\pi_m$ denotes the projection onto the free coordinates within $\theta_m$.
The limit for the estimation error then follows from the fact that 
$$
\partial_{\theta_m} lcml(\theta_0) = \pi_m \bar{U_n}
$$
where consequently
$$
\Lambda_m = -(\pi_m H \pi_m')^{-1} \pi_m.
$$
The asymptotic distribution of $\mu(\hat \theta_m)$ then follows from 
applying the Delta method. The formula for the asymptotic MSE then is obvious. 
\end{proof}

Therefore the asymptotics are for all models 
driven by the same random variable $U$. Thus it follows that the estimation errors for all models 
jointly are asymptotically normal. This allows the calculation of the 
asymptotic MSE also for the averaged model using weights $w_i$: 
\begin{eqnarray} \label{eq:MSE} 
N {\mathbb E} \left(\mu(\hat \theta^{MA})
-\mu  \left[ \begin{array}{c} \tau_0 \\ \gamma_0 + \delta/\sqrt{N} \end{array} \right] \right)^2 & \to &
\sum_{i,j =1}^M (\lambda_i  \delta \delta' 
\lambda_j ' + \partial_\theta\mu \Lambda_i J 
(\partial_\theta\mu \Lambda_j)')  w_i w_j \nonumber \\
& = & \sum_{i,j =1}^M F_{ij} w_i w_j = w'F w
\end{eqnarray} 
where the next to last equation defines the matrix $F \in {\mathbb R}^{M \times M}$ and the vector of weights $w = [w_1,...,w_M]' \in {\mathbb R}^{M}$. 
The optimal weights for estimating the value $\mu(\theta_0)$ then is provided by the solution to 
\begin{equation}
\label{eq:weights}
\hat{w}_{mse} = \arg \min_{w  \in {\mathbb R}^{M}, \sum_{m=1}^M w_m = 1} w'Fw.
\end{equation}

In practice the matrix $F$ is not known but contains  only quantities that can be estimated consistently except for $\delta$. In this respect it is customary to use the estimate $\hat \delta = \sqrt{N} (\hat \gamma - \gamma_0)$ from the wide model. This estimate is unbiased but not consistent. Note that contrary to the weights from the \ac{IC}-based approach the model averaging weights that result from (\ref{eq:weights}) might be either positive or negative. Without further restrictions the optimal weight vector can be calculated explicitly. 
\\
The matrix $F$ is not necessarily nonsingular, its rank is bounded by $d+1$. Conditions for nonsingularity of $F$ can be derived for particular settings, see \cite{charkhi2016}. This also implies that the estimation of $d+1$ models is sufficient in order to obtain minimum MSE weightings. This number typically is much smaller than the number of all possible combinations of regressor selections.  
\\
The main steps to compute an optimally averaged estimator may be summarized as:
\begin{enumerate}
\item{Decide on a list of candidate models}
\item{Estimate the wide model and obtain $\hat \delta = \sqrt{N}(\hat \gamma - \gamma_0)$ as an estimate of $\delta$.}
\item{Estimate its sensitivity $H$ as well as its variability matrix $J$.} 
\item{Estimate the remaining models, then compute the weights using (\ref{eq:weights}) and the averaged estimate.}
\end{enumerate}

This description leaves the question of the choice of $\mu(\theta)$ open, which in this context is usually called the focus parameter. A number of options in this respect are:
\begin{itemize}
	\item $\mu(\theta) = \theta_j$: extracting one parameter of interest.
    \item $\mu(\theta)$ being a prediction such as the conditional probability of a particular choice.
    \item Alternatively the sum of squared errors corresponding to all parameters can be used. 
\end{itemize}

In summary we have introduced two model averaging strategies, the first is rather ad-hoc but has the benefit that the weights are easily computed from either \ac{CLAIC} or \ac{CLBIC}. The other strategy involves additional computations but features a strong theoretic framing (optimality w.r.t asymptotic \ac{MSE} of some focus quantity).

\section{Model averaging: Comparison by simulation}

In this section we use a rather simplistic approach to explore the performance of model averaging in that we only average over the two different covariance structures involved in the test decisions of section~\ref{chap:simu}, where the models differ in 10 parameters.  

The setup was described in the introduction to section~\ref{chap:simu} but here our focus is on the estimation accuracy of the previously discussed averaging estimators when compared to model selection. The results are based on the \ac{MAE} over all Monte Carlo samples with regard to the -- arbitrarily chosen -- third parameter in $\textbf{b}$ ($\mu(\theta) = \theta_3$). So the difference between the models is in the covariance structure while we are interested in the impact of this difference on the estimates for a linear parameter. We first present the \ac{MAE} for the case of selecting the model using either \ac{CLAIC} or \ac{CLBIC} followed by the errors for different \ac{MA} methods. Note that as discussed previously the probability to select the restricted model is rather low and gets lower for rising $\alpha$ such that this \ac{MAE} is mainly driven by the unrestricted model. 

In Table~\ref{tab:covav} we see that as expected the \ac{MAE} for all methods gets smaller for rising sample sizes because the underlying estimation is getting better. Furthermore we see that the errors of model selection, which are presented in the first two rows, are also getting smaller for higher values of $\alpha$ which is due to the fact that in those cases both \ac{IC}s always select the unrestricted model. 
\\
The next two rows present the \ac{MAE} for an averaging estimator which is based on an \ac{IC} (see (\ref{eq:AICw})). We observe that this method is performing worse than model selection for the two smaller sample sizes but has lower errors for all but the highest values of $\alpha$ for sample size 1000. However, it is clearly visible that the error is not improving for rising values of $\alpha$ even though the weighting decision should get easier. 
\\
The results for an averaged estimator where the weights are chosen to minimize the asymptotic \ac{MSE} are given in the last row. Again we observe better performance for larger samples but the results also show that this averaging estimator is superior to the \ac{IC}-based \ac{MA} estimators. We further observe that this model averaging estimator outperforms model selection for small values of $\alpha$. It is important to recapitulate the results from Table~\ref{tab:cov} which show that the restricted model is wrongly almost never selected by \ac{CLAIC}/\ac{CLBIC}. This limited simulation exercise points out that model averaging -- especially when the weights are chosen to be \ac{MSE} optimally -- might provide a solution to this problem.

\begin{table}[H]
\centering
\caption{Covariance Structure: Mean Absolute Error (MAE) of various selection and averaging methods for $\textbf{b}_3$}
\input{tab_av_cov.tex}
\label{tab:covav}
\end{table}

\section{Model Averaging: Empirical Example on Mobility Motifs}

\label{case}
In this section we illustrate the use of the most promising of the previously discussed model averaging methods (asymptotically \ac{MSE} optimal weights) using data from the \ac{MOP}.  
The \ac{MOP} is a panel mobility survey with a rotating sample, keeping responding 
households in the sample for three consecutive years. Each year households are 
asked to record trips for all members of the household for one randomly assigned 
week. The corresponding trip diary conforms to the KONTIV design and 
collects trip start and end time and location, transport means, distances covered and trip purposes (for more details on the data generation and characteristics  see \citep{zumkeller2009}). 
For this paper we use data from 2013. 

Based on the trip diary for each person we compile mobility motifs 
\citep{schneider2013}: Here a motif is a graph containing nodes and directed 
edges, compare Figure~\ref{fig:mot}. The nodes represent locations people visited, the edges movements between 
edges. The significance of the motifs is seen in the fact that in different data 
sets in different cities it has been observed that from the large number of 
possible motifs with up to six nodes only relatively few (\citep{schneider2013} list 17 motifs covering more than 90\% of all day observations) turn 
out to occur frequently with also the frequency of occurrence being remarkably 
similar across studies. 

Beside these empirical facts motifs are of interest also for simulation models of 
mobility behaviour. Often such models only cover a single trip, while clearly trip chains 
exist and hence choices for one trip depend on the other trips within a day. 
Most often such dependencies are not or only partly taken into account (see \citep[p. 219ff]{cascetta2009}). In this respect it might be postulated that people in a first step choose one of the possible motifs in order to decide on the various locations to be visited and also the sequence of visits to these locations. And only in a second step the actual trips are planned in detail. Such activity plans also lie at the heart of 
some mobility simulation models such as the ones implemented in MATSim \citep{Axhausen2016}. 

For this case study, we  limit our analysis to the workdays of the year 2013. The motifs were computed from the trip diaries of 2369 participants and in this data set the 15 most common motifs (as depicted in Figure~\ref{fig:mot}) account for 92.6 percent of all choices and, therefore, we summarized the remaining 7.4 percent as 'other'.  

In this case study we investigate the choice of the actual motif on a workday based on underlying sociodemographic characteristics of the persons. 
\cite{cascetta2009} argues that it is mainly the occupational status of an individual that influences the individual trip-chaining. 
We will explore this hypothesis in this case study. To do so we consider a limited set of exploratory variables: a dummy which indicates when a person is not full-time employment, 
a gender dummy as well as dummies for four age groups (10-17, 18-25, 26-60, 61 and older).\footnote{Note that in Germany the legal age to acquire a (full) drivers license is 18. Furthermore, the compulsory school attendance ends around the age of 18 in most of the federal states. According to official figures the mean age of retirement in Germany was 61.9 in 2016 (see \citep{DRV2016}).} Those variables, which are included in the model as fixed effects, are summarizes in Table~\ref{tab:des}. 

\begin{table}[t]
\centering
\caption{Descriptive Statistics for the 2013 travel diary data set}
\input{tab_descript.tex}
\label{tab:des}
\end{table}

\begin{landscape}
\begin{figure}%
{\caption{The 15 most common motifs in the MOP data set. Note that the relative frequencies are w.r.t those 15 and not all occurring 824 motifs.}\label{fig:mot}}%
\begin{tabular}{@{}r@{}} 
\includegraphics[width=1.5\textwidth]{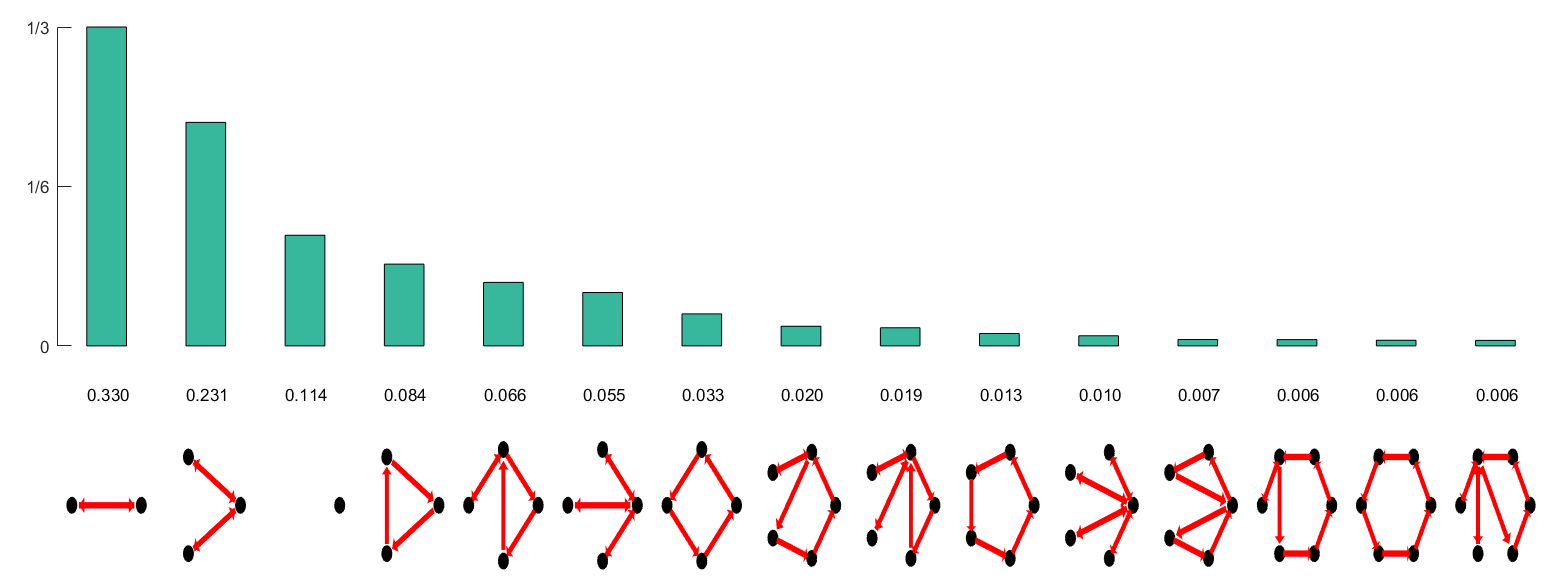}\\
\end{tabular}
\end{figure}
\end{landscape}

We explore the effect those variables have separately for each potential choice. As we observe five repeated choices for each participant we model the \ac{ASC} as random effects and the correlation between those random effects allows us to explore the substitution pattern between different motifs. In order to ensure identification the coefficients for the "other"-choice are fixed to one as is the corresponding variance. That leaves the model with 90 linear and  -- because the covariance parameters are estimated using the corresponding Cholesky decomposition ($\Omega = LL'$) -- with up to 120 covariance parameters.

Following through the steps outlined in section~\ref{chap:MA} we will start by specifying the list of candidate models. For this modeling exercise the specification of the covariance is of major concern because prior research suggests that mobility patterns are pretty stable during the work week and that, therefore, some substitution patterns might not be relevant (see \citep{schneider2013}). Furthermore, \cite{schneider2013} discuss that substitution in general can be explained by a rule-based system where several substitutions do not occur at all.  Instead of selecting a covariance structure, which could be done using the methods we discussed previously, we estimate several possible specifications ranging from the diagonal RE-specification without any correlation to the model featuring an unstructured covariance matrix.\footnote{Note that in order to ensure identification of the unstructured covariance we constrained all correlations between the 'other'-option and the remaining motif to zero.} In detail we consider four models,

\begin{itemize}
\item \textit{Unrestricted} correlations between the \ac{ASC}s of the 15 motifs
\item \textit{Block1}, substitution of the fist seven (simple) motifs to the other (more complex) motifs but no correlation within those blocks.
\item \textit{Block2}, like Block1 but also correlation within the first block,
\item \textit{Diagonal RE}, only uncorrelated random effects for each motif,
\end{itemize}
All those models include the 90 linear parameters and at least the 15 variances of the random ASCs. Using the terminology of section~\ref{chap:MA} those parameters form $\tau$ and the last model is the wide model which includes all 210 parameters (see Table~\ref{tab:mot}). The estimation for all four models is set up similar to the procedure outlined in section~\ref{chap:simu} but we initialize the optimizer at random because the true values are obviously unknown. 

In Table~\ref{tab:mot} we have summarized the different  \ac{IC} values for the models. Comparison of \ac{CLAIC} and \ac{CLBIC} reveals $Block2$ as the preferred specification. However, we can also observe that selection by \ac{CLAIC} and \ac{CLBIC} would lead to different choices for the second best model. While \ac{CLAIC} hints towards the unrestricted model, we would select the $Diagonal$ model would we rely on the \ac{CLBIC}. Furthermore, we show the weights for \ac{MSE} optimal model averaging. We focused the \ac{MSE} on the 15 coefficients which indicate whether an individual is not full-time employed. The weights might appear strange on first inspection because the model with the smallest \ac{CLBIC}/\ac{CLAIC} gets a negative weight but from Table~\ref{tab:mot_coeff} it is clear that those weights lead to meaningful averaged coefficients. The averaged coefficients reside either in between the estimate for $Block2$ and the full model or $Block2$ and $Block1$, this would have been impossible by just selecting one model. 

\begin{table}[t]
\centering
\caption{Information criteria and asymptotically \ac{MSE}-optimal weights for the fitted models}
\input{tab_mot.tex}
\label{tab:mot}
\end{table}

Even though the main focus of this section is to illustrate the feasibility of model averaging for real data we will interpret the estimates but without going into too much detail. In order to facilitate interpretation we have added a plot which contains the motifs alongside the estimated averaged coefficients. We see that the estimates are plausible as individuals without full-time employment seem to (I) stay at home more often and (II) favor motifs with a hub-structure, returning home in between visits to different locations, over round-trips. An interesting results is that the two motifs with the highest coefficients represent distinct activity patterns, the stay-at-home motif (blue box) and the hub-motif with four non-home locations (green box). The model would most likely benefit from considering interactions of the non-full-time employment dummy with other variables, which might explain the reason for the occupational status, to further explore this paradox.

Another possibility is to use model averaging to recover information about the correlations between the ASCs. In this case we focus on the \ac{MSE} with respect to all linear parameters but for the case of this empirical example we will only look at the estimated correlation matrices. First, note that by shifting the focus we obtain different weights,
\begin{equation*}
\hat{w} = [1.5546, -1.8555, 1.2105, 0.0904],
\end{equation*}
where the first weight is for the $Diagonal$ model, followed by the $Block1$, $Block2$ and $Unresticted$ specification. We see that the $Unrestricted$ model has the smallest weight which might be interpreted as a penalty for the high standard errors of the estimate. When compared directly we see that the estimates of the correlation matrix of the random effects for the unrestricted model (Figure~\ref{fig:corr_full}) and the averaged estimate (Figure~\ref{fig:corr_av}) reveal similar but not identical patterns. The ordering of the motifs is identical to that depicted in Figure~\ref{fig:mot}, therefore, motifs with lower numbers are observed more often. The first insight from those heatmaps is that the third motif is special in that its correlation with all other motifs is low in comparison to the other common motifs. This is plausible because it is the stay-at-home motif and is observed for the estimates from the averaged as well as from the $Unrestricted$ model.

\begin{landscape}
{\tiny
\begin{table}[t]
\centering
\caption{Estimated and averaged coefficients for the no-full-time employment dummy}
\input{tab_mot_coeff.tex}
\label{tab:mot_coeff}
\end{table}
}

\begin{figure}%
\label{pl:motif2}
{\caption{Averaged coefficients for the no-full-time employment dummy}\label{fig:mot2}}%
\begin{tabular}{@{}r@{}} 
\includegraphics[width=1.65\textwidth]{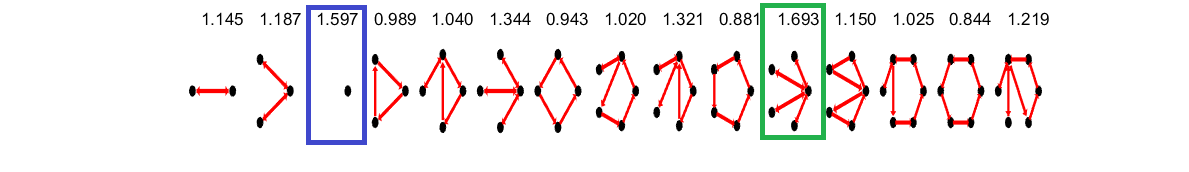}\\
\end{tabular}
\end{figure}
\end{landscape}

Another interesting finding is that there is a negative correlation between motif 7 (cyclic tour with 4 locations) and motif 9 (tour plus two separate trips) which is in line with the theory in \cite{schneider2013}. There is a difference between the estimates as the averaged estimate suggests that there also is a negative correlation between motif 7 and motif 10. Both estimates also reveal that there are in general low correlation between the blocks of the more common (motif 1-7) and the less common motifs. Finally it is worth noting that the number of observations for the motifs with a number larger than 10 are pretty small (see Figure~\ref{fig:mot}) which makes the estimate unreliable regardless of the method used for estimation.

This section showed that it is straightforward to apply the model averaging method developed in the previous sections to real world data. Note that we will present more detailed findings of a future analysis, which will involve more than one year of MOP data, elsewhere and that some computation times related to this empirical example are presented in the appendix.
\clearpage

\begin{figure}[t]%
{\caption{Correlation matrix of the ASCs for the $Unresticted$ model.}
\label{fig:corr_full}}%
\begin{tabular}{@{}r@{}} 
\includegraphics[width=\textwidth]{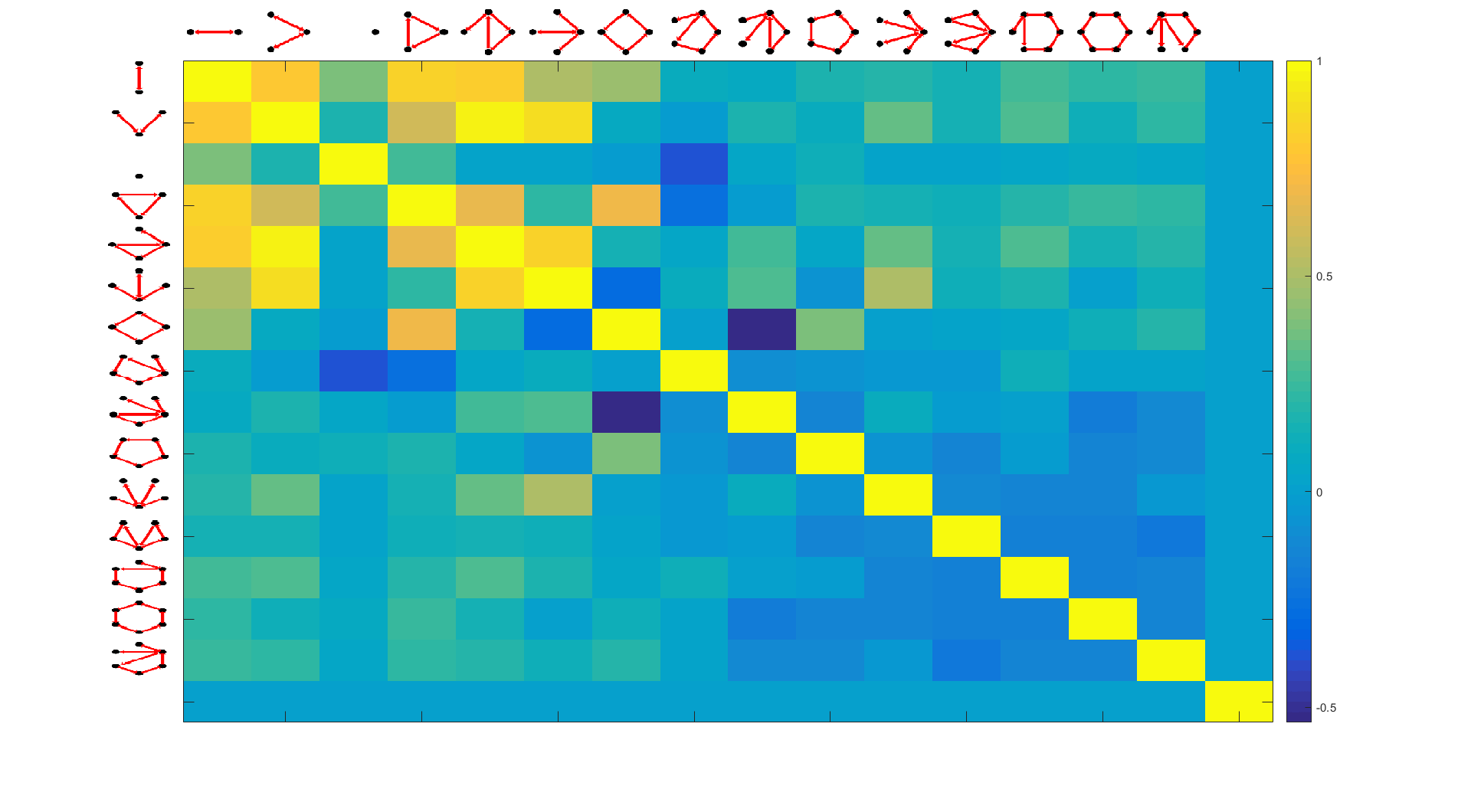}\\
\end{tabular}
\end{figure}
\begin{figure}[t]%
{\caption{Correlation matrix of the ASCs for the averaged model.}
\label{fig:corr_av}}%
\begin{tabular}{@{}r@{}} 
\includegraphics[width=\textwidth]{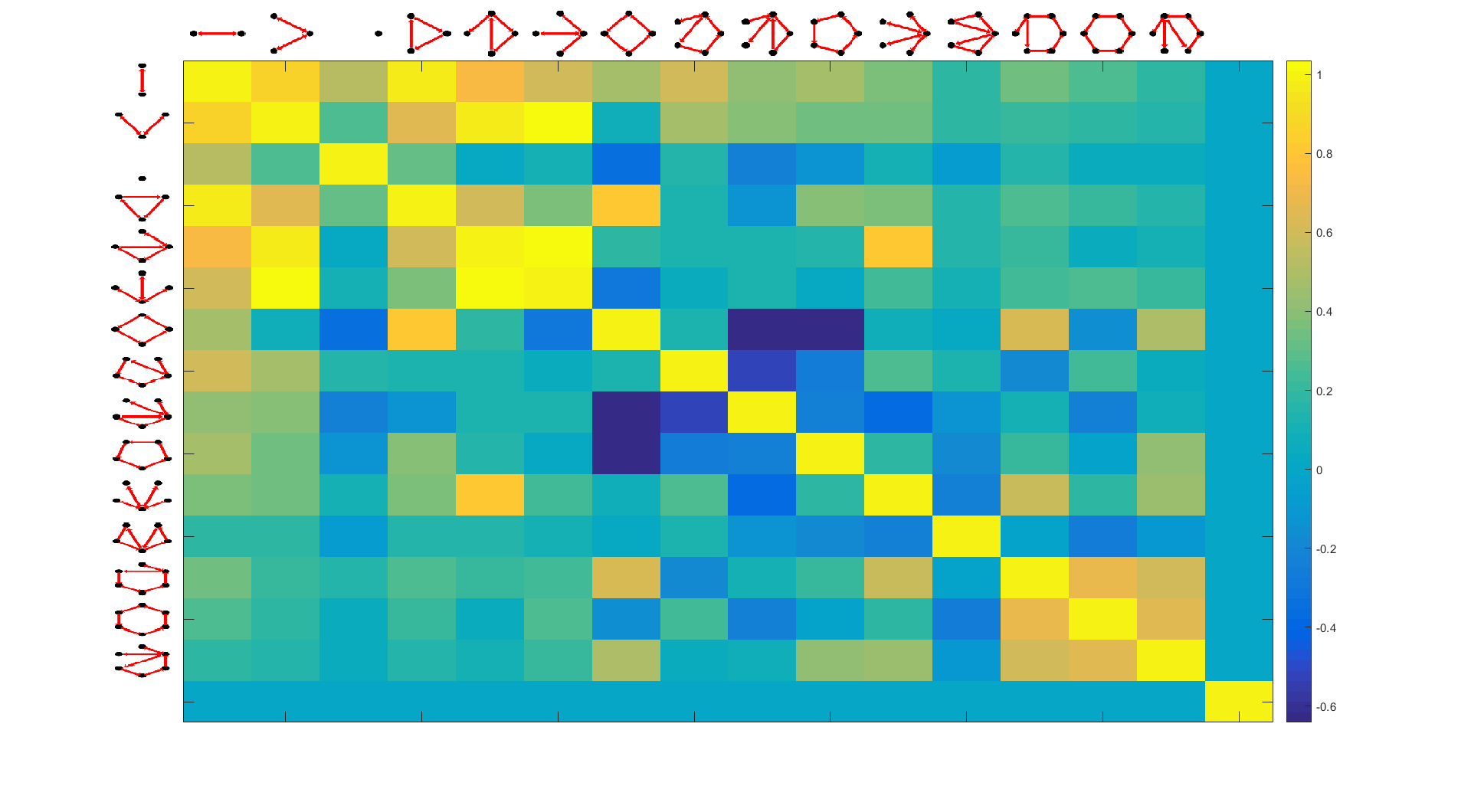}\\
\end{tabular}
\end{figure}

\clearpage

\section{Discussion and Conclusion}

In this paper we  investigated the performance of various model selection and model averaging methods. Initially we discussed that naive tests based on the ratio of two composite likelihoods are inapplicable and presented two alternative classes of tests which address this problem. First, there are correction based tests whose applicability to \ac{MACML} models was proposed in previous papers. Second, we proposed the use of an empirical likelihood test which to the best of our knowledge has not been used previously in the \ac{CML} context and for that matter in the \ac{MACML} literature. 

We assessed the performance (size and power) of those tests using a limited simulation exercise. A main result from section~\ref{chap:simu} is that the size of the tests ($cCLR_1 - cCLR_3$) that are based on a correction for $CLR$ are still -- at times severely -- inflated. Those results are in line with findings in the \ac{CML}-literature (see \citep{chandler2007} for a comparison of the naive and a corrected $CLR$, as well as \citep{geys1999}). Results in \cite{geys1999} hint to the problem that the size inflation gets more pronounced as the sample size rises but further analysis is needed to check whether this problem is specific to \ac{MACML}. 
\\
From the three tests that rely on correcting the $CLR$, the tests performances were similar for large sample sizes (1000) but the test based on moment matching (esp. $cCLR_2$) showed less severe size inflation when compared to the test of \cite{pace2011} ($cCLR_3$).  However, the results of our simulation suggest that the correction based tests in general are outperformed by the $EL$ test. The empirical size of this test was almost always closer to the nominal level when compared to the correction based test. At the same time the power of this test was equal or larger compared to $cCLR_1$ to $cCLR_3$ for most of the settings. As expected the $EL$ test shows inferior performance with regard to the power in small samples (but only for the covariance setting). In summary we think that it is safe to conclude that the $EL$ test is a viable choice for samples of size 1000 and upwards and very likely to outperformed correction-based alternatives in those settings.

In order to offer a more detail perspective on the performance of the various composite-likelihood-ratio-tests, it might be beneficial to reassess the performance of the tests using size-corrected-power (see for example \citep{davidson1984}). Furthermore, the problems which we observed for the correction based test warrant further investigation. It might be of special interest whether the problems are influenced by the choice of the analytic approximation to the \ac{MVNCDF} as there are potential alternatives to the \ac{SJ} approximation (see \cite{connors2014} or \cite{batram2016}). Some practical questions regarding the $EL$ test especially the sensitivity to violation of (\ref{eq:psi}) remain open.

We also discussed model selection strategy based on information criteria. Both of the popular \ac{IC}s have been adapted to the \ac{CML} framework in the form of the \ac{CLAIC} and \ac{CLBIC}. The main results of the simulations is that those criteria tend to select the largest model extremely often. Even though this problem was slightly less pronounced for the \ac{CLBIC} both criteria seem to be unsuited to discriminate between models especially with regard to differing covariance structures. The main problem here seems to be that the penalty terms, 
which are computed as the trace of (parts of) the estimated sandwich information matrix, is subject to estimation uncertainty and in general not large enough. That the penalty might not be large enough when it is estimated from the same data which also informed the likelihood, which is the first part of every \ac{IC} and in favor of large models by definition, is known in other parts of the model selection literature (see \cite{greven2010}).

Finally we introduce the idea of model averaging and show that the theoretic results of \cite{hjort2003} hold for \ac{CML} estimation. We also discussed the simpler \ac{IC}-based model averaging method. In a small simulation exercise we showed that the \ac{MA} estimators have the potential to outperform \ac{IC}-based model selection when concerned with estimation error. Especially for large sample sizes the asymptotically optimal model averaging showed promising results. However, \ac{MA} is mainly leveraged in a small simulation exercise and the empirical example and, therefore, more extensive simulation studies are needed in order to further the understanding of the performance of \ac{MA} when applied to \ac{MACML} models. 

Furthermore, the idea of focused model selection in the context of \ac{MNP}/\ac{MACML} models should be further explored. First, it is in principle possible to 'focus' on arbitrary quantities as long as those quantities depend on the model parameters. A natural example for \ac{MACML} models are the probabilities of different alternatives which are highly relevant for prediction. Second, it is also possible to derive the optimal \ac{MA} weights not only with respect to asymptotic \ac{MSE} but other metrics which might be more relevant to the research question at hand (see \cite{claeskens2006}). Those two options allow the researcher to tailor \ac{MA} estimators exactly to her research question. Finally, we have not even discussed the benefits of model averaging w.r.t to post-model-selection-inference (see \cite[199ff]{claeskens2008} and \cite{leeb2005}).

Based on  our theoretical discussion and the results of our simulations we would currently advise to use the empirical likelihood test to perform model selection for nested models. The inferior performance for small sample size should be a minor concern because the \ac{MACML} method is proposed for large data sets where traditional estimation methods for \ac{MNP} models like \ac{MSL} face computational problems. The \ac{CLAIC} and \ac{CLBIC} have the tendency to favor larger models especially when covariance structures are compared. Therefore, we suggest that researchers explore model averaging to combine the results from several models as shown in our empirical example.

\clearpage
\bibliography{lit}

\section*{Appendix: Computation times} 

In this appendix we provide tangible evidence regarding the computational performance of the methods discussed in this paper. All computations presented in this paper have been done on a laptop computer with an Intel i5-4215M with 2.6 Ghz (dual core) and 8 GB of RAM. All computation times are with respect to the unrestricted or the diagonal model presented in section~\ref{case} with estimation based on real world data, or the unrestricted model presented from section~\ref{chap:cov} fitted to simulated data. Given that the gradients are a by-product of the estimation we observe the timings given in the following table.
\begin{table}[H]
\centering
\caption{Computation times for various model specifications which were presented in the preceding sections}
\begin{tabular} 
 {l    *{3}{c} } 
 & $Unrestr.$ - real world  & $Diagonal$ - real world  & $Unrestr.$ - simulated \\ \hline \hline
no. individuals & 2369 & 2369 & 1000\\
no. of decision & 5 & 5 & 5\\
no of alterna. & 16 & 16 & 5\\
no of parameters & 210 & 105 & 20 \\
init. at truth & no & no & yes\\
\hline \hline
estimation ($\hat{\theta}$) & 11.2 hours & 1.1 hours & 10 seconds \\
$\hat{H}$ & 6.4 hours & 34 minutes & 15 seconds \\
$\hat{H_1}$ & 4.5 minutes & 2.3 minutes & 8 seconds \\
$F$/$\hat{w}$ (given $\hat{H}$/$\hat{H_1}$) & 1 second & 1 second & 1 second\\
$EL$ test & 80 seconds & 80 seconds & 2.5 seconds\\
\hline \hline
 \end{tabular} 

\end{table}
Note that there are technical differences between implementation of the likelihood function used for the optimization and that which is used to compute $\hat{H_1}$. Those computation times are related to the different model selection and model averaging methods as follows,

\begin{itemize}
\item Computation of either $\hat{H_1}$ (one likelihood evaluation) or $\hat{H}$ (numerical Hessian) is all that is needed for $cCLR_1$, $cCLR_2$, $cCLR_3$, $CLAIC$ and $CLBIC$.
\item The computation time for the $EL$ test is  due to the need to find $\psi$ for both models involved in the test.
\item Computing $F$ and solve for the four weights $w$ is very fast but there is also the need to compute $\hat{H_1}$ for the wide model, which needs to be added to the overall computation time for model averaging.
\end{itemize}

In summary and given that the computer system used for this task is an average business laptop and far from a high-performance system, the \ac{MACML} approach and the presented model selection and model averaging methods are highly appealing for practical use.

\end{document}

%% file: tab_linear.tex
\begin{tabular} 
 {l l |  *{6}{c} } 
n & $\beta$ & 0  &  0.1 & 0.2 & 0.3 & 0.4 & 0.5\\ \hline \hline 
300 & $CLR$& 0.304 &0.872 &0.996 &1.000 &1.000 &1.000 \\ 
& $cCLR$& 0.082 &0.600 &0.972 &0.998 &1.000 &1.000 \\ 
 & $EL$& 0.049 &0.533 &0.971 &0.998 &1.000 &1.000 \\ 
 & $CLAIC$& 0.971 &0.997 &1.000 &1.000 &1.000 &1.000 \\ 
 & $CLBIC$& 0.924 &0.989 &0.999 &1.000 &1.000 &1.000 \\ \hline \hline 
500 & $CLR$& 0.346 &0.944 &1.000 &1.000 &1.000 &1.000 \\ 
 & $cCLR$& 0.125 &0.792 &0.996 &1.000 &1.000 &1.000 \\ 
 & $EL$& 0.066 &0.760 &0.997 &0.999 &1.000 &1.000 \\ 
 & $CLAIC$& 0.966 &0.999 &1.000 &1.000 &1.000 &1.000 \\ 
 & $CLBIC$& 0.900 &0.996 &1.000 &1.000 &1.000 &1.000 \\ \hline \hline 
1000 & $CLR$& 0.375 &0.991 &1.000 &1.000 &1.000 &1.000 \\ 
 & $cCLR$& 0.159 &0.940 &0.999 &1.000 &1.000 &1.000 \\ 
 & $EL$& 0.064 &0.946 &1.000 &1.000 &1.000 &1.000 \\ 
 & $CLAIC$& 0.947 &0.999 &1.000 &1.000 &1.000 &1.000 \\ 
 & $CLBIC$& 0.863 &0.996 &1.000 &1.000 &1.000 &1.000 \\ \hline 
\hline 
 \end{tabular} 

%% file: tab_cov.tex
\begin{tabular} 
 {l l |  *{5}{c} } 
n & $\alpha$ & 0  &  0.1 & 0.2 & 0.3 & 0.4 \\ \hline \hline 
300 & $CLR$& 0.926 &0.938 &0.975 &0.983 &0.996  \\ 
 & $cCLR_1$& 0.077 &0.178 &0.415 &0.726 &0.924  \\ 
 & $cCLR_2$& 0.071 &0.171 &0.409 &0.723 &0.923   \\ 
 & $cCLR_3$& 0.108 &0.198 &0.434 &0.742 &0.926  \\ 
 & $EL$& 0.165 &0.253 &0.508 &0.774 &0.949  \\ 
 & $CLAIC$& 1.000 &1.000 &1.000 &1.000 &1.000  \\ 
 & $CLBIC$& 0.999 &0.999 &0.999 &1.000 &1.000  \\ \hline \hline 
500 & $CLR$& 0.905 &0.947 &0.984 &0.992 &0.999  \\ 
 & $cCLR_1$& 0.091 &0.231 &0.594 &0.907 &0.986  \\ 
 & $cCLR_2$& 0.090 &0.226 &0.587 &0.906 &0.984  \\ 
 & $cCLR_3$& 0.109 &0.254 &0.622 &0.908 &0.987 \\ 
 & $EL$& 0.099 &0.221 &0.549 &0.879 &0.989 \\ 
 & $CLAIC$& 0.999 &0.999 &1.000 &1.000 &1.000  \\ 
 & $CLBIC$& 0.999 &0.999 &0.997 &1.000 &1.000  \\ \hline \hline 
1000 & $CLR$& 0.929 &0.959 &0.985 &1.000 &1.000 \\ 
 & $cCLR_1$& 0.106 &0.339 &0.869 &0.993 &1.000 \\ 
 & $cCLR_2$& 0.104 &0.335 &0.867 &0.993 &1.000  \\ 
 & $cCLR_3$& 0.104 &0.354 &0.865 &0.993 &1.000 \\ 
 & $EL$& 0.078 &0.233 &0.794 &0.991 &1.000  \\ 
 & $CLAIC$& 0.998 &0.996 &1.000 &1.000 &1.000 \\ 
 & $CLBIC$& 0.994 &0.994 &0.998 &1.000 &1.000 \\ \hline 
\hline 
 \end{tabular} 

%% file: tab_av_cov.tex
\begin{tabular} 
 {l l |  *{5}{c} } 
n & $\alpha$ & 0  &  0.1 & 0.2 & 0.3 & 0.4 \\ \hline \hline 
300 & CLAIC Select.& 0.201 &0.203 &0.193 &0.177 &0.160  \\ 
& CLBIC Select. & 0.201 &0.205 &0.193 &0.177 &0.160  \\ \hline 
 & CLAIC Aver.& 0.225 &0.231 &0.234 &0.222 &0.221   \\ 
 & CLBIC Aver. & 0.225 &0.230 &0.234 &0.222 &0.221  \\ 
 & oMSE Aver.& 0.226 &0.219 &0.214 &0.196 &0.184 \\ \hline \hline 
500 & CLAIC Select.& 0.153 &0.151 &0.151 &0.143 &0.128  \\ 
& CLBIC Select. & 0.153 &0.151 &0.151 &0.143 &0.128 \\ \hline 
 & CLAIC Aver.& 0.158 &0.152 &0.156 &0.159 &0.154  \\ 
 & CLBIC Aver. & 0.158 &0.152 &0.156 &0.159 &0.154  \\ 
 & oMSE Aver.& 0.157 &0.149 &0.150 &0.147 &0.136  \\ \hline \hline 
1000 & CLAIC Select.& 0.121 &0.122 &0.120 &0.116 &0.102  \\ 
& CLBIC Select. & 0.121 &0.122 &0.119 &0.116 &0.102  \\ \hline 
 & CLAIC Aver.& 0.117 &0.110 &0.113 &0.116 &0.116   \\ 
 & CLBIC Aver. & 0.117 &0.111 &0.113 &0.116 &0.116  \\ 
 & oMSE Aver.& 0.116 &0.112 &0.113 &0.113 &0.106  \\ \hline 
\hline 
 \end{tabular} 

%% file: tab_descript.tex
\begin{tabular} 
 {l  l l}
 Variable & Value & Frequency \\ \hline \hline
 occupation & full-time employed & 819  \\
 & not fulltime employed & 1550 \\ \hline 
  gender & male & 1181  \\
 & female & 1188 \\ \hline 
  age & 10-17 & 173  \\
 &18-25 & 116 \\
 &25-61 & 1247 \\
 & 61+ & 833 \\ \hline 
 N & & 2369 \\ \hline \hline
 \end{tabular} 

%% file: tab_mot.tex
\begin{tabular} 
 {l    *{4}{c} } 
Model & no. of parameters  & CLAIC & CLBIC &weights  \\  \hline \hline
$Diagonal$ &        105  &     193298 &     195941 & 0.7580 \\ 
$Block1$ &        161  &     193045 &     196258& -0.0328 \\ 
$Block2$ &        182  & \textbf{    192235} &     \textbf{195673} & -0.2320 \\ 
$Unrestricted$ &        210  &     192485 &     196189& 0.5068 \\ 
\hline \hline
 \end{tabular} 

%% file: tab_mot_coeff.tex
\begin{tabular} 
 {l    *{15}{c} } 
Model & Mot1 & Mot2 & Mot3 & Mot4 & Mot5 & Mot6 & Mot7 & Mot8 & Mot9 & Mot10 & Mot11 & Mot12 & Mot13 & Mot14 & Mot15   \\  \hline \hline
$Unrestricted$ &  1.165 &  1.205&  1.612&  0.994&  1.050&  1.375&  0.934&  1.050&  1.353&  0.873&  1.709&  1.136&  1.041&  0.865&  1.224 \\ 
$Block2$ &  1.157 &  1.199&  1.587&  0.985&  1.036&  1.368&  0.928&  1.013&  1.311&  0.880&  1.723&  1.173&  1.012&  0.852&  1.224  \\ 
$Block1$ &  1.140 &  1.164&  1.597&  0.971&  1.025&  1.323&  0.948&  1.009&  1.339&  0.864&  1.695&  1.123&  0.945&  0.806&  1.203  \\ 
$Diagonal$ &  1.135 &  1.178&  1.585&  0.984&  1.031&  1.329&  0.945&  0.997&  1.298&  0.886&  1.692&  1.166&  1.007&  0.831&  1.216 \\ \hline  
$Averaged$ &  1.145 &  1.187&  1.597&  0.989&  1.040&  1.344&  0.943&  1.020&  1.321&  0.881&  1.693&  1.150&  1.025&  0.844&  1.219  \\ 
\hline \hline
 \end{tabular} 